\documentclass[a4paper, 12pt]{article}
\usepackage{fullpage}
\usepackage{amsmath, amsthm, amssymb, mathrsfs, graphicx}
\usepackage{ifthen, url, multirow}
\usepackage[usenames]{color}

\usepackage[sort&compress]{natbib}
\bibpunct{(}{)}{;}{a}{}{,} 

\newcommand{\A}{\mathscr{A}}

\newcommand{\FF}{\mathbb{F}}
\newcommand{\MM}{\mathbb{M}}
\newcommand{\FFbar}{\overline{\mathbb{F}}}
\newcommand{\Y}{\mathscr{Y}}
\newcommand{\U}{\mathscr{U}}

\newcommand{\RR}{\mathbb{R}}         

\newcommand{\E}{\mathsf{E}}

\newcommand{\prob}{\mathsf{P}}

\newcommand{\eps}{\varepsilon}

\newcommand{\grad}{\nabla}

\newcommand{\lmarg}[1]{L_{#1}^{\text{\sc b}}}
\newcommand{\lprml}[1]{L_{#1}^{\text{\sc m}}}
\newcommand{\lprof}[1]{L_{#1}^{\text{\sc p}}}
\newcommand{\hlmarg}[1]{\widehat L_{#1}^{\text{\sc b}}}

\newcommand{\dpp}{\mathrm{DP}}

\newcommand{\iid}{\overset{\text{\tiny iid}}{\,\sim\,}}

\theoremstyle{plain}
\newtheorem{theorem}{Theorem}

\theoremstyle{remark}

\theoremstyle{definition}

\theoremstyle{remark}

\theoremstyle{definition}
\newtheorem{assumption}{Assumption}

\theoremstyle{definition}
\newtheorem*{PRalgorithm}{PR Algorithm}

\begin{document}

\title{Semiparametric inference in mixture models with predictive recursion marginal likelihood}
\author{
Ryan Martin \\
Department of Mathematical Sciences \\
Indiana University--Purdue University Indianapolis \\
\url{rgmartin@math.iupui.edu} \\
\mbox{} \\
Surya T. Tokdar \\
Department of Statistical Science \\
Duke University \\
\url{tokdar@stat.duke.edu}
}
\date{April 7, 2011}

\maketitle

\begin{abstract}
Predictive recursion is an accurate and computationally efficient algorithm for nonparametric estimation of mixing densities in mixture models.  In semiparametric mixture models, however, the algorithm fails to account for any uncertainty in the additional unknown structural parameter.  As an alternative to existing profile likelihood methods, we treat predictive recursion as a filter approximation to fitting a fully Bayes model, whereby an approximate marginal likelihood of the structural parameter emerges and can be used for inference.  We call this the predictive recursion marginal likelihood.  Convergence properties of predictive recursion under model mis-specification also lead to an attractive construction of this new procedure.  We show pointwise convergence of a normalized version of this marginal likelihood function.  Simulations compare the performance of this new marginal likelihood approach that of existing profile likelihood methods as well as Dirichlet process mixtures in density estimation.  Mixed-effects models and an empirical Bayes multiple testing application in time series analysis are also considered.

\smallskip

\emph{Keywords and phrases.} Density estimation; Dirichlet process mixture; empirical Bayes; filtering algorithm; marginal likelihood; martingale; mixed effects model; multiple testing; profile likelihood.
\end{abstract}

\section{Introduction}
\label{S:intro}

Consider data $Y_1,\ldots,Y_n$ modeled as independent draws from a common, nonparametric mixture distribution with density 
\begin{equation} 
\label{simp mix mod}
m_f(y) = \int p(y \mid u) f(u) \,d\mu(u), \quad y \in \Y,
\end{equation}
where $(y,u) \mapsto p(y \mid u)$ is a known kernel on $\Y \times \U$ and $f$ is an unknown mixing density in $\FF$, the set of densities with respect to a $\sigma$-finite Borel measure $\mu$ on $\U$.  \cite{nqz} introduced the following stochastic algorithm, called predictive recursion, to estimate $f$ and $m_f$. 

\begin{PRalgorithm}
Choose an initial estimate $f_0 \in \FF$ of $f$, and a sequence of weights $w_1,\ldots,w_n \in (0,1)$.  For $i=1,\ldots,n$, compute the following:
\begin{align}
f_{i}(u) & = (1-w_i)f_{i-1}(u) + w_i \frac{p(Y_i \mid u) f_{i-1}(u)}{\int p(Y_i \mid u') f_{i-1,}(u') \,d\mu(u')}, \quad u \in \U, \label{eq:recursion} \\
m_{i}(y) & = m_{f_{i}}(y) = \int p(y \mid u) f_{i}(u) \,d\mu(u), \quad y \in \Y. \label{eq:recursion.m}
\end{align}
Return $f_{n}$ and $m_{n}$ as the estimates of $f$ and $m_{f}$, respectively. 
\end{PRalgorithm}

The algorithm's strengths include its fast computation and its unique flexibility to estimate the mixing density with respect to any user-specified dominating measure $\mu$.  Predictive recursion also has a close connection to Dirichlet process mixture models; see \citet{nqz}, \citet{quintana2000}, \citet{newtonzhang}, \citet{newton02}, and Section~\ref{S:prml}. \cite{tmg} show that when $Y_1,\ldots,Y_n$ are generated independently from a density $m$ which equals $m_{f^\star}$ for some $f^\star \in \FF$, the resulting estimates $f_n$ and $m_n$ converge as $n \to \infty$, respectively and in appropriate topologies, to $f^\star$ and $m_{f^\star}$; see also \citet{ghoshtokdar} and \citet{martinghosh}.  \citet{mt-rate} show that if $m$ does not equal $m_f$ for any $f \in \FF$, then the estimates still converge, but now the limits are characterized by the minimizer $f^\star$ of the Kullback--Leibler divergence $K(m, m_f) = \int \log \{ m(y) / m_f(y) \} m(y) \, dy$. The minimizer exists and is unique under certain conditions. An upper bound on the rate of this convergence is also available.

In statistical applications, however, an exact description of the kernel $p(y \mid u)$ is rarely available.  It is more common to use a family of kernels $p(y \mid \theta,u)$ indexed by a parameter $\theta \in \Theta$ and model $Y_1,\ldots,Y_n$ as independent draws from a semiparametric mixture 
\begin{equation}
\label{eq:mixture}
m_{f,\theta}(y) = \int p(y \mid \theta,u) f(u) \,d\mu(u), 
\end{equation}
where both $\theta$ and $f$ are unknown.  A frequently encountered example of this is density estimation with mixtures of Gaussian kernels, where $p(y \mid \theta, u) = N(y \mid u, \theta^2)$ with $\theta$ playing the role of a bandwidth.  A related formulation is in the linear mixed effects model $Y_i = U_i + X_i'\beta + \sigma \eps_i$, where $\theta = (\beta,\sigma)$ is unknown, and the density $f$ of the random effect $U_i$ is not restricted to a parametric family.  While $\theta$ is more like a nuisance parameter in the density estimation problem, it takes center stage in the mixed effect model.  In either case, predictive recursion fails to provide any statistical analysis for $\theta$.

\citet{taonewton1999} counter this shortcoming by embedding predictive recursion in a profile likelihood framework.  At any given $\theta$, one runs the predictive recursion algorithm with kernel $p(y \mid \theta,u)$ and a suitable initial guess $f_{0,\theta}$ to recursively compute $f_{i, \theta}$ and $m_{i, \theta}$ for $i=1,\ldots,n$. The final update $m_{n, \theta}$ is then plugged in to give the following profile likelihood in $\theta$:
\begin{equation}
\label{eq:prof}
\lprof{n}(\theta) = \prod_{i=1}^n m_{n,\theta}(Y_i). 
\end{equation}
\citet{taonewton1999} maximize this profile likelihood to estimate $\theta$.  Such a plug-in approach does not account for the lack of precision in estimating the mixing density $f$.  In the density estimation setting, the profile likelihood may be maximized at the zero bandwidth, completely ignoring the extreme variability of the estimates of $f$ at small bandwidths.  Such undesirable behavior can be avoided by imposing a penalty on the estimate of $f$.  But a general framework along these lines is yet to emerge, particularly for problems where inference on $\theta$ is the main focus.  

In this paper we demonstrate that predictive recursion's close connection with the Bayesian paradigm offers a rich alternative to the plug-in approach.  By viewing it as an approximation to fitting a fully Bayesian model on $(\theta,f)$, it is natural to ask whether it can also provide an approximation to the marginal likelihood for $\theta$ as defined by the Bayesian model.  In Section~\ref{S:prml} we show that such an approximation is indeed available and of the form
\begin{equation}
\label{eq:Lprml}
\lprml{n}(\theta) = \prod_{i=1}^n m_{i-1,\theta}(Y_i). 
\end{equation}
The approximate marginal likelihood $L_n(\theta)$, which we call the predictive recursion marginal likelihood, appears to inherit the intrinsic Ockham's razor properties \citep{jefferys} of the original Bayesian formulation.  That is, the $\theta$ values for which the conditional prior on $f$ is more spread out automatically receive greater penalty.  

In Section~\ref{S:asymptotics} we show that if $Y_1,\ldots,Y_n$ are independent samples from a density $m$, then $\log \lprml{n}(\theta)$ equals $-n \inf_{f \in \FF} K(m, m_{f, \theta})$ plus a quantity that grows slower than $n$.  A consequence of this is a convergence property of the maximum predictive recursion marginal likelihood estimate 
\begin{equation}
\label{eq:mprmle}
\hat \theta_n = \arg\max_{\theta \in \Theta} \lprml{n}(\theta)
\end{equation}
that follows from an argument similar to that of \citet{wald1949}. Specifically, if $\Theta$ is finite, then $\hat \theta_n$ converges to $\theta^\star$ as $n \to \infty$, where the limit is characterized by the minimizer $(\theta^\star, f^\star)$ of $K(m, m_{f,\theta})$ over $\Theta \times \FF$.  Our simulation studies suggest that similar results should hold for compact $\Theta$ as well, but so far a proof has eluded us.

An exact sampling distribution for $\hat\theta_n$ in \eqref{eq:mprmle} is not available.  Therefore, for inference on $\theta$ we estimate the standard error of $\hat \theta_n$ via the curvature of $\lprml{n}$ at its maxima.  This is motivated by the interpretation of the predictive recursion marginal likelihood as an approximate Bayesian marginal likelihood for which Laplace approximation applies \citep{tierney.kadane.1986}.

Several examples are presented in Section~\ref{S:examples}.  For density estimation, our simulations indicate that $\lprml{n}$ closely approximates Bayes Dirichlet process mixture marginal likelihood, whereas $\lprof{n}$ is more sporadic, in some cases concentrating on the boundary of the parameter space.  Applications to interval estimation in random-intercept regression models and multiple testing in mixtures of autoregressive process models are also given.

\section{Approximation to the Dirichlet process mixture marginal likelihood}
\label{S:prml}

As noted in \citet{nqz} and \citet{newton02}, the updating scheme \eqref{eq:recursion} has a close connection with the posterior updates in a Bayesian formulation when $f$ is modeled by a Dirichlet process prior. In this section we further explore this connection to establish $L_n(\theta)$ as an approximation to the marginal likelihood of $\theta$ as defined by such a Bayesian formulation. 

To be precise, consider the following extension of the mixture model \eqref{eq:mixture}:
\begin{equation}
\label{eq:Fmixture}
m(y) = m_{F,\theta}(y) = \int p(y \mid \theta,u) \,dF(u), 
\end{equation}
where $F$ is an unknown probability measure on $\U$, not necessarily dominated by $\mu$. Consider a Bayesian formulation
\begin{equation}
\label{eq:bayes.model}
Y_1,\ldots,Y_n \mid (F, \theta) \iid m_{F, \theta}, \quad F \mid \theta \sim \Pi_\theta, \quad \theta \sim \Gamma, 
\end{equation}
where $\Pi_\theta$ is, for each $\theta \in \Theta$,  a probability distribution over the space of probability measures $F$, and $\Gamma$ is a probability distribution on $\Theta$.  The posterior distribution $\Gamma_n$ of $\theta$ given the $n$ observations can be written as $d\Gamma_n(\theta) \propto \lmarg{n}(\theta) \,d\Gamma(\theta)$, where
\[ \lmarg{n}(\theta) = \int \Bigl\{  \prod_{i=1}^n m_{F,\theta}(Y_i) \Bigr\} \,d\Pi_\theta(F)\]
is the marginal likelihood of $\theta$ obtained by integrating out $F$ from \eqref{eq:bayes.model}.  For every $\theta \in \Theta$, let $\Pi_{i,\theta}$ denote the conditional posterior distribution of $F$ given $\theta$ and the first $i$ observations, i.e., $d\Pi_{i,\theta}(F) \propto \bigl\{\prod_{j = 1}^i m_{F, \theta}(Y_j) \bigr\} \, d\Pi_\theta(F)$.  Then by linearity of $m_{F,\theta}$ and Fubini's theorem, 
\begin{align} 
\lmarg{n}(\theta) & = \prod_{i=1}^n \int m_{F, \theta}(Y_i) \, d\Pi_{i - 1,\theta}(F) = \prod_{i=1}^n \int p(Y_i \mid \theta,u) \,dF(u)  d\Pi_{i - 1,\theta}(F) \notag \\
& = \prod_{i=1}^n \int p(Y_i \mid \theta,u) \,d F_{i-1,\theta}(u), \label{eq:lmargsimp}
\end{align}
where $F_{i, \theta} = \int F \, d\Pi_{i ,\theta}(F)$ is the conditional posterior mean of $F$ given $(Y_1,\ldots,Y_i,\theta)$.  

Now consider the special case where $\Pi_\theta = \dpp(\alpha_0, F_{0,\theta})$, the Dirichlet process distribution with precision parameter $\alpha_0 > 0$ and base measure $F_{0, \theta}$ \citep{ferguson1973,ghoshramamoorthi}.  Assume that the base measures $F_{0, \theta}$ are all absolutely continuous with respect to $\mu$, admitting densities $f_{0,\theta} = dF_{0,\theta}/d\mu \in \FF$. It follows from the Polya urn representation of a Dirichlet process \citep{blackwellmacqueen} that
\begin{equation}
\label{eq:dpupdate}
dF_{1,\theta}(u) = \frac{\alpha}{\alpha+1} d F_{0,\theta}(u) + \frac{1}{\alpha+1} \frac{p(Y_1 \mid \theta,u) \,dF_{0,\theta}(u)}{\int p(Y_1 \mid \theta,u') \,dF_{0,\theta}(u')}. 
\end{equation}
Therefore, $F_{1,\theta}$ is absolutely continuous with respect to $\mu$ and the density  $dF_{1,\theta} / d\mu \in \FF$ is identical to the predictive recursion output $f_{1,\theta}$ based on the single observation $Y_1$, with initial guess $f_{0,\theta}$, kernel $p(y \mid \theta,u)$ and weight $w_1 = 1 / (1 + \alpha_0)$. Consequently $\lmarg{2}(\theta) = \lprml{2}(\theta)$ as can be verified by comparing \eqref{eq:Lprml} and \eqref{eq:lmargsimp}.

This analogy, however, does not carry over to $\lmarg{i}(\theta)$ and $\lprml{i}(\theta)$ for $i \geq 3$. For $i = 3$, the relevant conditional posterior mean $F_{2,\theta}$ does not admit a representation as in \eqref{eq:dpupdate} in terms of $F_{1,\theta}$ and $p(Y_2 \mid \theta,u)$ because the conditional posterior distribution $\Pi_{1,\theta}$ is no longer a Dirichlet process distribution, but rather a mixture of Dirichlet processes \citep{antoniak1974}. 

To remedy this, consider an approximation to the Bayesian model, where we successively replace $\Pi_{i, \theta}$ with $\widehat \Pi_{i, \theta} = \dpp(\alpha_i, \widehat F_{i, \theta})$ where $\widehat F_{0, \theta} = F_{0,\theta}$ and 
\[\widehat F_{i, \theta} = \frac{\int F \, m_{F,\theta}(Y_i) \, d\widehat\Pi_{i - 1, \theta}(F) }{ \int m_{F, \theta}(Y_i) \, d\widehat \Pi_{i - 1, \theta}(F)}, \quad i \geq 1, \]
is what one would obtain for $E(F \mid Y_1, \cdots, Y_i, \theta)$ if the conditional posterior of $F$ given $(Y_1,\ldots,Y_{i - 1},\theta)$ was indeed $\widehat \Pi_{i - 1, \theta}$. These successive replacements can be thought of as a dynamic, mean preserving, filter approximation to the original Bayesian model.  Note that every $\widehat F_{i,\theta}$ remains absolutely continuous with respect to $\mu$ and satisfies the recursion
\[ d\widehat F_{i,\theta}(u) = \frac{\alpha_{i - 1}}{1 + \alpha_{i - 1}} d\widehat F_{i - 1, \theta}(u) + \frac{1}{1 + \alpha_{i - 1}} \frac{p(Y_i \mid \theta,u)\,d\widehat F_{i - 1, \theta}(u)}{\int p(Y_i \mid \theta,u') \, d\widehat F_{i-1,\theta}(u')}. \]
This, coupled with the initial condition $\widehat F_{0,\theta} = F_{0,\theta}$, implies that the densities $d \widehat F_{i,\theta} / d\mu$ are precisely the $f_{i, \theta}$ that result from predictive recursion applied to the observations $Y_1,\ldots,Y_n$, with initial guess $f_{0,\theta}$, kernel $p(y \mid \theta,u)$ and weights $w_i = 1/ (1 + \alpha_{i - 1})$.  Therefore the corresponding approximation $\hlmarg{n}(\theta) = \prod_{i = 1}^n \int p(Y_i \mid \theta,u) \, d\widehat F_{i-1,\theta}(u)$ of $\lmarg{n}(\theta)$ is exactly $\lprml{n}(\theta)$.

For every $\theta \in \Theta$, the quantity $\prod_{i=1}^n m_{i-1,\theta}(Y_i)$ indeed defines a joint probability density for $(Y_1,\ldots,Y_n)$ which admits $\lprml{n}(\theta)$ as an exact likelihood function for $\theta$.  However it is unknown whether this joint density corresponds to any exchangeable hierarchical model on the $Y_i$'s, thus making it somewhat unsuitable to use it for statistical analysis. For this reason, we do not focus on studying $\lprml{n}(\theta)$ from the perspective of the joint model for which it is an exact likelihood function. We are instead interested in studying it as an inferential tool when $Y_i$'s are generated independently from a common density $m$ which may or may not be a mixture as in \eqref{eq:mixture}.  

\section{Asymptotic theory}
\label{S:asymptotics}

\subsection{Notation and preliminaries}

For $\FF$, the set of all densities on $\mathscr{U}$ with respect to $\mu$ as in Section~\ref{S:intro}, let $\FFbar$ be its closure with respect to the weak topology.  With a slight abuse of notation, the elements of $\FFbar$ are also denoted by $f$, although they need not admit a density with respect to $\mu$.  For each $\theta$, let $\MM_\theta = \{m_{f,\theta} : f \in \FFbar\}$. \citet{mt-rate} show that the predictive recursion estimates $m_{n,\theta}$ converge, for each fixed $\theta$, to the best mixture density in $\MM_\theta$, if the following assumptions hold.

\begin{assumption}
\label{as:distribution}
Observations $Y_1,Y_2,\ldots$ are independent with a common density $m$, and $K(m,m')$ is finite for all $m' \in \MM = \bigcup_{\theta \in \Theta} \MM_\theta$.  
\end{assumption}
\begin{assumption}
\label{as:weights}
The weight sequence satisfies $\sum_n w_n = \infty$ and $\sum_n w_n^2 < \infty$.
\end{assumption}
\begin{assumption}
\label{as:precompact}
The set $\FFbar$ is compact with respect to the weak topology.  
\end{assumption}
\begin{assumption}
\label{as:continuity}
The mapping $u \mapsto p(y \mid \theta,u)$ is bounded and continuous for all $(y,\theta)$.
\end{assumption}
\begin{assumption}
\label{as:bound1}
For each $(\theta_1,\theta_2)$ pair, there exists $A = A(\theta_1,\theta_2) < \infty$ such that 
\[ \sup_{u_1,u_2} \int \Bigl\{ \frac{p(y|\theta_1,u_1)}{p(y|\theta_2,u_2)} \Bigr\}^2 m(y)\,dy \leq A. \]
\end{assumption}

Assumption \ref{as:weights} is standard in the literature on stochastic approximation algorithms, of which predictive recursion is a special case \citep{martinghosh}, and it holds if $w_n$ decays like $n^{-\gamma}$ for $\gamma \in (1/2,1]$.  Assumption~\ref{as:precompact} is satisfied if, for example, $\mathscr{U}$ is compact and $\mu$ is Lebesgue measure. The more demanding Assumption~\ref{as:bound1}, holds for many standard kernels $p(y \mid \theta,u)$, such as those arising from Gaussian or other exponential family distributions, whenever $\mathscr{U}$ is compact and $m$ admits a moment-generating function on $\mathscr{U}$. 

Define the mapping 
\begin{equation}
\label{eq:KLopt}
K^\star(\theta) = \inf\{K(m,m_{f,\theta}): f \in \FFbar\}, \quad \theta \in \Theta, 
\end{equation}
the smallest Kullback--Leibler divergence over $\MM_\theta$.  Attainment of the infimum in \eqref{eq:KLopt} follows from Assumptions~\ref{as:precompact} and \ref{as:continuity}; see \citet{mt-rate}, Lemma~3.1.  Let $a_n = \sum_{i=1}^n w_i$ denote the partial sums of the weight sequence $\{w_n: n \geq 1\}$.  For two real sequences $\{\alpha_n\}$ and $\{\beta_n\}$, we write $\alpha_n = O(\beta_n)$ if $\alpha_n/\beta_n$ is bounded, and $\alpha_n = o(\beta_n)$ if $\alpha_n / \beta_n \to 0$.  Then \citet{mt-rate} prove a version of the following theorem.  

\begin{theorem}
\label{thm:mt}
Under Assumptions~\ref{as:distribution}--\ref{as:bound1}, $K(m,m_{n,\theta}) \to K^\star(\theta)$ almost surely as $n \to \infty$ for each fixed $\theta$. In addition to Assumptions~\ref{as:distribution}--\ref{as:bound1}, if the infimum in \eqref{eq:KLopt} is attained in the interior of $\FF$, and if $\sum_n a_n w_n^2 < \infty$, then $K(m,m_{n,\theta}) - K^\star(\theta) = o(a_n^{-1})$ almost surely. 
\end{theorem}

For weights that satisfy $w_n = O(n^{-\gamma})$, the extra condition, $\sum_n a_n w_n^2 < \infty$, in the second part of the theorem holds if and only if $\gamma \in (2/3,1]$.  Therefore, the best available rate in this case is $o(n^{-1/3})$ almost surely.  But \citet{mt-rate} argue that this rate is conservative.

\subsection{Main results}

Write $\ell_n = \log \lprml{n}$ and define the following normalization:
\begin{equation}
\label{eq:Kn}
K_n(\theta) = \frac1n \sum_{i=1}^n \log \frac{m(Y_i)}{m_{i-1,\theta}(Y_i)} = -\frac{\ell_n(\theta) - \ell_{0n}}{n},
\end{equation}
where $\ell_{0n}= \sum_{i = 1}^n \log m(Y_i)$ is the log joint density of $Y_1, \ldots, Y_n$.  We will show that the main result of Theorem~\ref{thm:mt} holds with $K_n(\theta)$ in place of $K(m,m_{n,\theta})$. 

\begin{assumption}
\label{as:bound2}
For each $\theta$, there exists $B = B_\theta < \infty$, such that the density $m_{f,\theta} \in \MM_\theta$ at which the infimum in \eqref{eq:KLopt} is attained satisfies
\[ \int \Bigl\{ \log \frac{m(y)}{m_{f,\theta}(y)} \Bigr\}^2 m(y) \,dy \leq B. \]
\end{assumption}

If the mixture model is correctly specified, i.e., $m \in \MM$, then Assumption~\ref{as:bound2} is a consequence of Assumption~\ref{as:bound1} and Jensen's inequality; see the argument for \eqref{eq:jensen} in Appendix~1.  But this assumption seems reasonable even if $m \not\in \MM$, as it assures that the proposed mixture model \eqref{eq:mixture} is, in a certain sense, close enough to the truth.  

\begin{theorem}
\label{thm:emp-KL}
Under Assumptions~\ref{as:distribution}--\ref{as:bound2}, $K_n(\theta) \to K^\star(\theta)$ almost surely as $n \to \infty$ for each fixed $\theta$.  In addition to Assumptions~\ref{as:distribution}--\ref{as:bound2}, if the infimum in \eqref{eq:KLopt} is attained in the interior of $\FF$, and if $\sum_n a_n w_n^2 < \infty$, then $K_n(\theta) - K^\star(\theta) = O(a_n^{-1})$ almost surely. 
\end{theorem}

\begin{proof}
See Appendix~1.
\end{proof}

From Theorem~\ref{thm:emp-KL} we can conclude that $\ell_n(\theta) = -nK^\star(\theta) + \ell_{0n} + o(n)$.  Therefore for any $\theta_1, \theta_2$ with $K^\star(\theta_1) \neq K^\star(\theta_2)$, the difference $\ell_n(\theta_1) - \ell_n(\theta_2)$ grows linearly in $n$. Such linear growth in log likelihood differences is the building block of Wald's famous proof of consistency of maximum likelihood estimators.  In our case, we have the following.  

\begin{theorem}
\label{thm:m-prml-e}
Suppose $\Theta$ is a finite set and Assumptions 1--6 hold. Let $\Theta^\star = \{\theta': K^\star(\theta') = \inf_\Theta K^\star(\theta)\}$. Then $\hat \theta_n \in \Theta^\star$ almost surely for all sufficiently large $n$.  In particular, if $\Theta^\star = \{\theta^\star\}$, then $\hat \theta_n \to \theta^\star$ almost surely.
\end{theorem}

\begin{proof}
Let $\Theta^\dagger = \Theta \setminus \Theta^\star$.  By Theorem~\ref{thm:emp-KL} and the finiteness of $\Theta$, $\sup_{\Theta^\dagger}\ell_n(\theta) - \ell_n(\theta^\star) = -n \{\inf_{\Theta^\dagger} K^\star(\theta) - \inf_{\Theta} K^\star(\theta)\} + o(n) \to -\infty$, almost surely. On the other hand, $\ell(\hat \theta_n) - \ell_n(\theta^\star) \geq 0$ by definition of $\hat \theta_n$ and thus $\hat \theta_n \not\in\Theta^\dagger$ almost surely for all large $n$. 
\end{proof}

While Theorem~\ref{thm:m-prml-e} is limited to finite $\Theta$, our empirical results Section~\ref{SS:limit} suggest that the convergence can be extended to a compact $\Theta$. Toward this, we must show that $K_n$ is either uniformly equi-continuous or convex, a task to be taken up in a future work.

It is unclear whether results similar to Theorems~\ref{thm:emp-KL} and \ref{thm:m-prml-e} hold for the profile likelihood $\lprof{n}$. Our proofs of these two theorems are based on martingale theory and rely on the fact that the summand $\log m_{i - 1, \theta}(Y_i)$ in $\lprml{n}(\theta)$ is measurable with respect to the $\sigma$-algebra generated by $Y_1, \ldots, Y_i$. This does not hold for the summands $\log m_{n,\theta}(Y_i)$ in $\log \lprof{n}(\theta)$.

\subsection{Numerical illustration}
\label{SS:limit}

Suppose $Y_1,\ldots,Y_n$ are modeled as independent observations from the mixture density $m_{f,\sigma}(y) = \int_0^1 p(y \mid \sigma, u) f(u) \,du$, where $p(y \mid \sigma,u)$ is a $N(u, \sigma^2)$ kernel.  The true model, however, is $Y_i = 0.5 + 0.1 Z_i$, where $Z_1,\ldots,Z_n$ is a random sample from the Student-t distribution with $5$ degrees of freedom.  The true density $m(y)$ underlying the shifted and scaled Student-t model cannot equal $m_{f,\sigma}(y)$ for any $f$ with support $[0, 1]$ and $\sigma > 0$, because the tails of $m_{f,\sigma}(y)$ decay exponentially while those of $m(y)$ decay polynomially.  But despite this model mis-specification, $K_n(\sigma) \to K^\star(\sigma)$ pointwise in $\sigma$ by Theorem~\ref{thm:emp-KL}.  

We approximate $K(m, m_{f,\sigma})$ with its Gaussian quadrature form 
\[ \int m(y) \log \frac{m(y)}{m_{f,\sigma}(y)} dy \approx \sum_{r = 1}^R b_r m(y_r) \log \frac{m(y_r)}{\sum_{j = 1}^J a_j p(y_r | \sigma, u_j)f(u_j)}, \]
which is then numerically optimized over $f(u_1),\ldots,f(u_J)$ to obtain an approximation to $K^\star(\sigma)$. Here $\{(a_j,u_j): j=1,\ldots,J\}$ are the Legendre node-weight pairs of order $J = 101$ for the interval $[0, 1]$, and $\{(b_r,y_r): r=1,\ldots,R\}$ are the same for $[-0.5, 1.5]$ with $R = 101$.  Figure~\ref{fig:limit} shows $K_n(\sigma)$ and the limit $K^\star(\sigma)$ for three choices of $n$, each replicated with 100 independent data sets from the shifted and scaled Student-t distribution.  In addition to showing pointwise convergence, Figure~\ref{fig:limit} also gives an indication of the conjectured convexity of $K_n(\sigma)$.
 
\begin{figure}
\begin{center}
\includegraphics[scale=0.95]{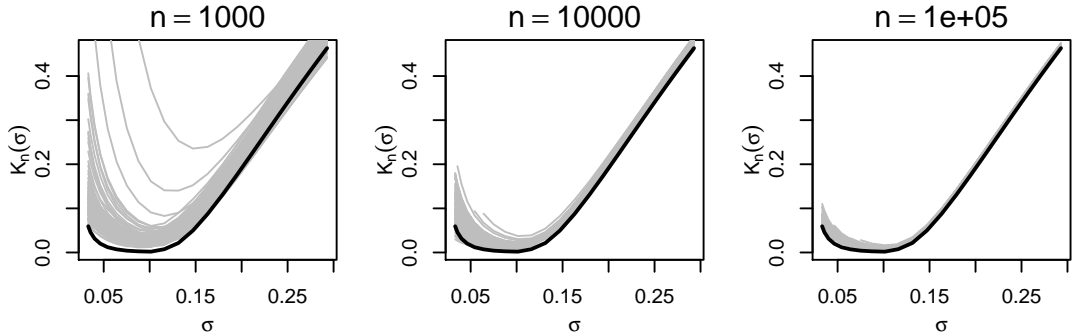}
\end{center}
\caption{Gray lines show $K_n(\sigma)$ for 100 independent data sets from the mis-specified mixture model in Section~\ref{SS:limit}.  Black lines show an approximation of the pointwise limit $K^\star(\sigma)$. }
\label{fig:limit}
\end{figure}

\section{Examples}
\label{S:examples}

\subsection{Density estimation}
\label{SS:density}

Consider density estimation where $p(y \mid \sigma, u) = N(y \mid u, \sigma^2)$ is the Gaussian kernel with bandwidth $\sigma$ playing the role of $\theta$ and the mixing density $f$ assumed to have support $[0, 1]$.  For the Bayes approach, the distribution function $F$ is modeled as a draw from a Dirichlet process distribution with unit precision parameter and a uniform base measure on $[0,1]$.  For predictive recursion, we take $f_0$ to be a uniform density on $[0,1]$, matching the Dirichlet process base measure, and set $w_i = (i+1)^{-2/3}$.  Figure~\ref{fig:sims} shows $\lmarg{n}(\sigma)$, $\lprml{n}(\sigma)$ and $\lprof{n}(\sigma)$ for 12 data sets of size $n = 50$ simulated from the Gaussian mixture model $m_{f,\sigma}(y) = \int p(y \mid \sigma,u) f(u) \,du$ for several $(\sigma, f)$ pairs.  The four different mixing distributions were chosen to capture various shapes and characteristics, including discrete, continuous, and a mixture of each.  The importance sampling technique of \citet{tmg} is used to evaluate $\lmarg{n}(\sigma)$; see also \citet{maceachernclydeliu}.    

\begin{figure}
\begin{center}
\includegraphics[scale=0.97]{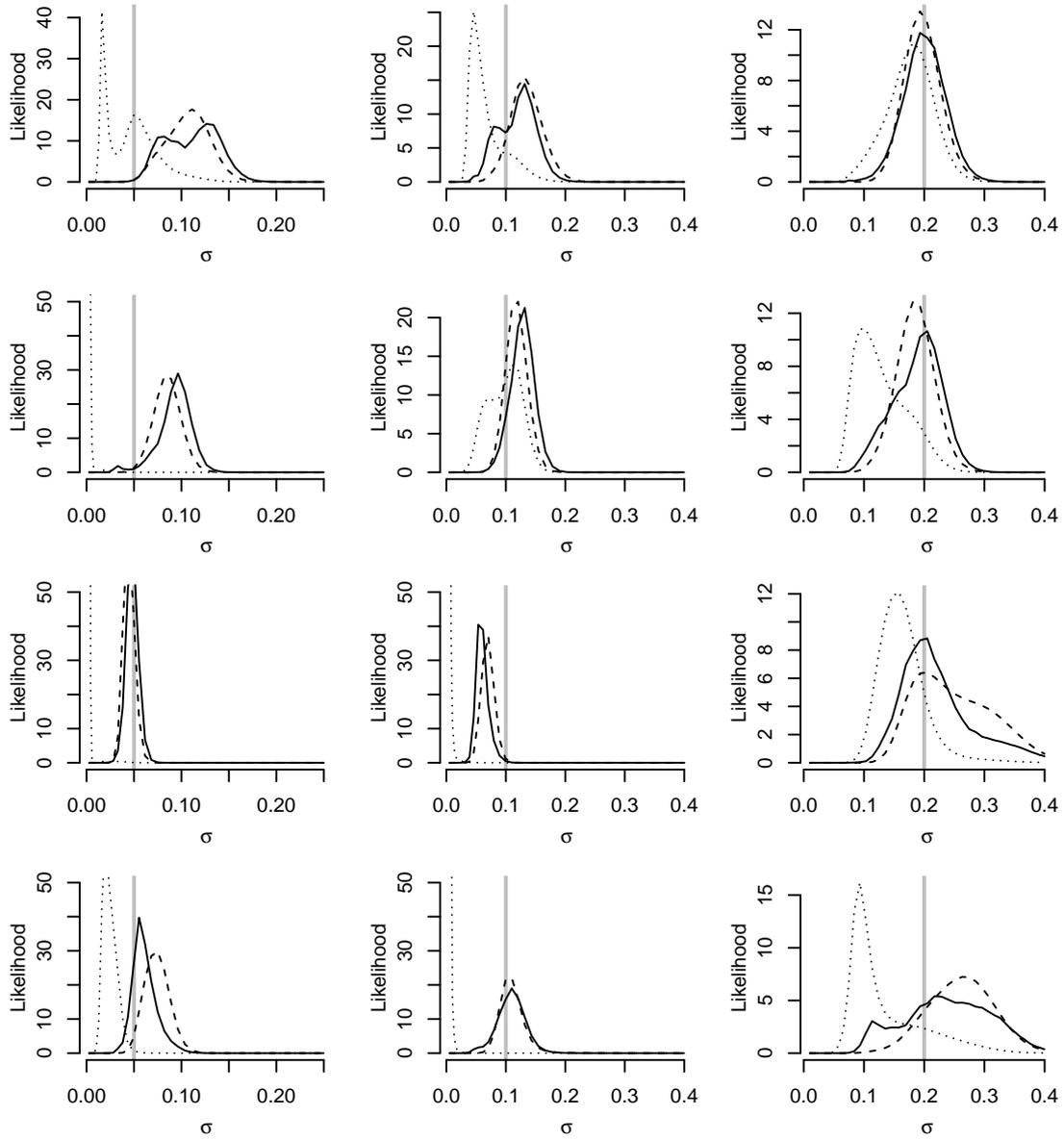}
\end{center}
\caption{Solid lines show $\lmarg{n}(\sigma)$, dashed lines show $\lprml{n}(\sigma)$, and dotted lines show $\lprof{n}(\sigma)$ for the density estimate example in Section~\ref{SS:density}.  These curves have been normalized to integrate to 1.  The true value of $\sigma$ is marked by a vertical gray line.  The figure's four rows correspond to four different mixing distributions.  The first is a $\text{Beta}(2,6)$ density; the second is a $\text{Beta}(10,30)$ density; the third is an equal mixture of point masses at $1/4$ and $3/4$; the fourth is an equal mixture of $\text{Beta}(2,6)$ and a point mass at $3/4$. }
\label{fig:sims}
\end{figure}

Figure~\ref{fig:sims} shows the normalized marginal likelihoods for the three different methods.  It is clear that $\lprml{n}(\sigma)$ closely approximates $\lmarg{n}(\sigma)$ in all cases, while $\lprof{n}(\sigma)$ deviates arbitrarily, sometimes peaking at $\sigma = 0$. The approximation of $\lmarg{n}(\sigma)$ by $\lprml{n}(\sigma)$ is robust against the smoothness and skewness properties of the underlying mixing distribution $f$.  This is quite striking because the Dirichlet process formulation views $f$ as a discrete distribution while predictive recursion is designed to recover smooth densities.  We also note that both the approximate marginal and profile likelihood calculations are orders of magnitude faster than those for Dirichlet process mixture; see \citet{tmg} for a comparison of run times.

\subsection{Random-intercept regression models}
\label{SS:mixed}

In this section we study two regression models along the lines of \citet{taonewton1999}. In each case, we consider data on a response $Y$ and a vector of predictors $X \in \RR^d$ for $n$ subjects each with $r$ replicates. In our first study, $Y$ is a continuous variable and is linked to the predictors through the random-intercept linear regression model:
\begin{equation}
\label{eq:lmm}
Y_{ij} \mid X_{ij} \sim N(U_i + X_{ij}'\beta, \sigma^2), 
\end{equation}
where $i = 1,\ldots, n$ indexes subjects, and $j = 1,\ldots, r$ indexes replicates.  We assume that the $Y_{ij}$'s are conditionally independent across both $i$ and $j$. The subject-specific intercepts, the $U_i$'s, are taken to be independent draws from a probability density $f$ with respect to the Lebesgue measure on an interval $\mathscr{U} = [a,b]$. Write $\theta = (\beta,\sigma^2)$ as the unknown parameter of interest.  A related semiparametric Bayes model appears in \citet{bush.maceachern.1996}.  

To cast this regression model as the mixture model \eqref{eq:mixture}, we assume $X_i = (X_{i1}, \ldots, X_{ir})$ are stochastic, sampled independently from a density $g$ over $\RR^d \times \cdots \times \RR^d$.  Then \eqref{eq:lmm} says that $(X_i, Y_i)$, where $Y_i = (Y_{i1}, \cdots, Y_{ir})$, are independent observations from a density $m = m_{f,\theta}$ as in \eqref{eq:mixture} with a kernel
\begin{equation}
\label{eq:lmm.kernel}
p(x, y \mid \theta, u) = p_x(y \mid \theta,u) g(x),
\end{equation}
where the conditional density of $Y$ given $x$ is 
\[ p_x(y \mid \theta, u) = \frac{1}{(2\pi\sigma^2)^{r/2}} \exp\Bigl\{-\frac{1}{2\sigma^2} \sum_{j = 1}^r (y_j - x_j'\beta - u)^2 \Bigr\}. \]
With this setup, $\theta$ can be estimated by maximizing the predictive recursion marginal or profile likelihood.  The predictor density $g$ need not be estimated: it drops out from the updating equation \eqref{eq:recursion}, and so does not affect $f_{i, \theta}$.  Consequently, $m_{i,\theta}(x, y) = m_{i,\theta,x}(y)g(x)$, where $m_{i,\theta,x}(y) = \int p_x(y|\theta,u)f_{i,\theta}(u) \,du$ does not involve $g$, and hence $\log \lprml{n}(\theta)$ incorporates $g$ only through an additive constant.  Also note that by factoring $m(x, y) = m_x(y)g(x)$, one can write $K(m, m_{f,\theta}) = \int K(m_x, m_{f, \theta,x}) g(x)\,dx$ where $m_{f, \theta,x}(y) = \int p_x(y|\theta, u) f(u)\,du$. Thus the $(\theta^\star, f^\star)$ that characterizes the limiting asymptotic properties of $\lprml{n}(\theta)$ minimizes an average Kullback--Leibler divergence $K(m_x, m_{f,\theta,x})$ of the conditional densities weighted by $g(x)$.

Evaluation of $\lprml{n}(\theta)$ requires a single pass through the predictive recursion algorithm for each $\theta$, which is computationally inexpensive, even for large $n$.  Optimization, however, requires several evaluations of $\lprml{n}$ and its gradient.  With computational efficiency in mind, we present a version of predictive recursion that also produces $\grad \lprml{n}$ as a by-product, with no substantial increase in computational cost; see Appendix~2.  This gradient algorithm, coupled with any packaged optimization routine, makes for fast semiparametric estimation of $\theta$.

For inference on $\theta$, the exact sampling distribution of $\hat\theta$ in \eqref{eq:mprmle} will not be available in general, so some sort of approximation is needed.  Here we propose a curvature based approximation, where the covariance matrix of $\hat\theta$ is estimated by the inverse Hessian $H = \{-\grad^2 \ell_n(\hat \theta)\}^{-1}$, readily obtained from the output of the optimization routine.  A stipulated $100(1 - \alpha)\%$ confidence interval for $\theta_j$ can be obtained by taking $\hat \theta_j \pm z_{\alpha / 2} h_{jj}$, where $z_q$ denotes the upper $q^{\text{th}}$ quantile of the standard normal distribution, and $h_{jj}$ is the $j^{\text{th}}$ diagonal element of $H$. 

Table~\ref{tab:lmm} reports the performance of the predictive recursion-based marginal and profile likelihood estimates averaged over 500 datasets generated from model \eqref{eq:lmm} with $n = 50$ or $500$, $r = 4$, $d=2$, $\beta = (2,5)'$, and $\sigma = 2$.  The covariates $X_{ij1}$'s are independent $N(0, 1)$ and $X_{ij2} = J_i + 0.1 Z_{ij}$, with $J_i$'s satisfying $\prob(J_i=0) = \prob(J_i=1) = 0.5$ and $Z_{ij}$'s independent $N(0,1)$. In this case, $X_{ij1}$ is a within-subject covariate while $X_{ij2}$, which is roughly constant in $j$, acts like a subject-specific covariate.  Three choices of $f$ are considered: $N(0, 2^2)$, a shifted exponential distribution with rate 0.5 and support $(-2,\infty)$, and a discrete uniform supported on $\pm 2$.  Each choice of $f$ has mean zero and variance 4. See \citet{taonewton1999} for more details on these choices. For comparison to a parametric model fit, we also include a likelihood-based method that assumes the mixing distribution is Gaussian with unknown parameters. 

Performance of each estimate is measured by the root mean-square error in estimating each component of $\theta$. We also include average coverage of a stipulated 95\% confidence interval for each estimating method constructed as described above.  Predictive recursion is run with $w_i = (1 + i)^{-2/3}$, $\U = \overline{Y} \pm 3 S_Y$ where $\overline{Y}$ and $S_Y$ are the mean and standard deviation of the full data $\{Y_{ij}\}$, and  $f_0$ the uniform density over $\U$.  In this example, the predictive recursion-based marginal and profile likelihood methods perform similarly.  Both are competitive with the parametric Gaussian method when $f$ is indeed Gaussian, and are better for non-Gaussian $f$. The relative similarity between the marginal and the profile likelihood methods, unlike what we observed in the density estimation example before, can be explained by noting that for model \eqref{eq:lmm}, the data is informative about the parameter $\sigma$ due to availability of replicates.  We note that although the coverage probabilities are generally close to the stipulated 95\% level, there are some noticeable differences.  First, the marginal and profile likelihood coverage probabilities for $\beta_2$ are off the mark in the small $n$ case.  That $X_{ij2}$ is partially confounded with the group structure is one potential explanation for this phenomenon.  Second, in estimating $\sigma$ for the discrete uniform model, which lies in the boundary of $\FF$, the coverage falls dangerously low, even for large $n$.  For such boundary cases, bootstrap confidence intervals might be more appropriate; see Section~\ref{S:discuss}. 

In our second example, we retain much of the above setting but consider a binary $Y$ for which model \eqref{eq:lmm} is adapted to the a random-intercept logistic regression model:
\begin{equation}
\label{eq:glmm}
\text{logit} \bigl\{ \prob(Y_{ij} = 1) \bigr\} = U_i + X_{ij}'\beta,
\end{equation}
where $U_i$'s are independent draws from a uniform density $f$ on $\U = [-8,8]$ and $\theta = (\beta_1, \beta_2)$ is unknown. This model corresponds to \eqref{eq:mixture} through \eqref{eq:lmm.kernel} with an appropriate choice of $p_x(y|\theta,u)$.  The last four columns of Table~\ref{tab:lmm} reports the performance of maximum likelihood estimates of $\theta$ based on the predictive recursion marginal and profile likelihood and the parametric Gaussian set, with same choices for $\beta_1$, $\beta_2$, $n$, $r$, $f$, $w_i$, and $\{X_{ij}\}$ as in our first study.  For $n=500$, all methods perform similarly; for $n=50$ the proposed marginal likelihood approach is better in terms of both estimation accuracy and coverage. 

\begin{table}
\begin{center}
{\small
\begin{tabular}{cc|ccc|ccc|cc|cc}
& & \multicolumn{6}{c}{Study I} & \multicolumn{4}{|c}{Study II}\\[2pt]
$f$ & Method & \multicolumn{3}{c}{RMSE} & \multicolumn{3}{|c|}{Coverage} & \multicolumn{2}{c}{RMSE} & \multicolumn{2}{|c}{Coverage}\\ 
 & & $\beta_1$ & $\beta_2$ & $\sigma$ & $\beta_1$ & $\beta_2$ & $\sigma$ & $\beta_1$ & $\beta_2$ & $\beta_1$ & $\beta_2$\\ 
\hline
&& \multicolumn{6}{|c}{$n = 50$} & \multicolumn{4}{|c}{$n = 50$}\\
\hline
\multirow{3}{*}{Gaussian} & Gaussian & 0.16 & 0.60 & 0.12 & 95 & 95 & 94 & 1.02 & 2.48 & 88 & 86 \\ 
  & Marginal & 0.16 & 0.68 & 0.12 & 95 & 86 & 94 & 0.61 & 1.58 & 95 & 91 \\ 
  & Profile & 0.16 & 0.71 & 0.12 & 94 & 89 & 92 & 0.65 & 1.78 & 95 & 90\\ 
\hline
\multirow{3}{*}{Exponential} & Gaussian & 0.17 & 0.61 & 0.11 & 93 & 94 & 95 & 0.86 & 2.13 & 90 & 91\\
  & Marginal & 0.17 & 0.52 & 0.11 & 94 & 91 & 95 & 0.64 & 1.66 & 96 & 96 \\ 
  & Profile & 0.17 & 0.50 & 0.12 & 94 & 92 & 94 & 0.67 & 1.76 & 96 & 96 \\ 
\hline
\multirow{3}{*}{Uniform} & Gaussian & 0.15 & 0.58 & 0.11 & 96 & 95 & 95 & 0.94 & 2.42 & 87 & 88\\ 
  & Marginal & 0.14 & 0.36 & 0.11 & 96 & 97 & 93 & 0.82 & 1.92 & 92 & 92 \\ 
  & Profile & 0.14 & 0.34 & 0.12 & 96 & 97 & 90 & 0.89 & 2.14 & 92 & 93 \\ 
\hline
&& \multicolumn{6}{|c}{$n = 500$} & \multicolumn{4}{|c}{$n = 500$}\\
\hline
\multirow{3}{*}{Gaussian} & Gaussian & 0.05 & 0.18 & 0.04 & 96 & 95 & 95 & 0.17 & 0.43 & 88 & 85 \\ 
  & Marginal & 0.05 & 0.19 & 0.04 & 96 & 94 & 95 & 0.15& 0.39 & 93 & 93 \\ 
  & Profile & 0.05 & 0.19 & 0.04 & 96 & 94 & 95 & 0.16 & 0.41 & 93 & 95\\ 
\hline
\multirow{3}{*}{Exponential} & Gaussian & 0.05 & 0.18 & 0.04 & 93 & 96 & 95 & 0.15 & 0.41 & 90 & 87\\ 
  & Marginal & 0.05 & 0.15 & 0.04 & 95 & 94 & 95 & 0.15 & 0.32 & 96 & 96  \\ 
  & Profile & 0.05 & 0.14 & 0.04 & 95 & 95 & 94 & 0.15 & 0.34 & 95 & 95\\ 
\hline
\multirow{3}{*}{Uniform} & Gaussian & 0.05 & 0.19 & 0.04 & 96 & 94 & 94 & 0.18 & 0.50 & 88  & 84\\ 
  & Marginal & 0.05 & 0.11 & 0.05 & 94 & 95 & 80 & 0.20 & 0.52 & 91 & 89 \\ 
  & Profile & 0.05 & 0.11 & 0.05 & 95 & 94 & 84 & 0.21 & 0.54 & 90 & 86 \\ 
\hline
\end{tabular}
}
\end{center}
\caption{Comparison of the predictive recursion-based marginal and profile likelihood methods along with a parametric Gaussian model for parameter estimation in the random-intercept regression models in Section~\ref{SS:mixed}.  RMSE is root mean-square error.  Coverage probabilities are multiplied by 100.  Studies~I and II denote the linear and logistic models, respectively.  }
\label{tab:lmm}
\end{table}

The marginal and profile likelihood methods used above inherit the order-dependence characteristic of predictive recursion.  That is, $\lprml{n}(\theta)$ and $\lprof{n}(\theta)$ depend on the order in which the data enter the recursion \eqref{eq:recursion}.  The order is not important asymptotically, but finite sample behaviors may be sensitive to the specific order used.  A simple remedy to this is to apply predictive recursion on $M$ random permutations of the data and average the $M$ resulting marginal and profile likelihood functions prior to optimization.  Permutation-averaging in the two regression examples gave only marginally better results compared to what is shown in Table~\ref{tab:lmm}.

\subsection{Multiple testing in time series}
\label{SS:testing}

Let $Y = \{Y(t):t=1,2,\ldots,T\}$ denote a discrete-time stochastic process under a first-order autoregressive model; i.e., $Y(0)$ is normal with mean 0 and variance $\sigma^2 (1-\phi)^{-1}$, and 
\[ Y(t) = \xi + \phi Y(t-1) + \sigma \eps(t), \quad \eps(t) \sim N(0,1), \quad t = 1,2,\ldots,T, \]
for $\xi \in \RR$, $\phi \in (-1,1)$ and $\sigma^2 > 0$.  In other words, $Y$ has a $T$-dimensional normal distribution with mean $\xi 1_T = (\xi,\ldots,\xi)' \in \RR^T$ and $T \times T$ covariance matrix $\Sigma_u$ of the form 
\[ \Sigma_{u,jk} = \frac{\sigma^2}{1-\phi} \, \phi^{|j-k|}, \quad u = (\sigma^2,\phi) \in \U \subset (0,\infty) \times (-1,1). \] 
For a sample $Y_1,\ldots,Y_n$ of similar processes, consider a mixture of simple first-order autoregressive models, namely, 
\begin{equation}
\label{eq:armixture}
Y_i \mid (\xi_i,U_i) \sim N(\xi_i 1_T, \Sigma_{U_i}), \quad U_i \sim f(u), \quad \xi_i \mid \theta \sim \theta \langle 0 \rangle + (1-\theta) N(0,1),
\end{equation}
where independence is assumed throughout, and $\langle 0 \rangle$ denotes a degenerate distribution at 0.   It shall be further assumed that $\theta$ is close to 1, so that $Y_1,\ldots,Y_n$ are sparse in the sense that most have mean zero.  The goal is to identify those processes with non-zero mean trajectories.  

Model \eqref{eq:armixture} is similar to that which appears in \citet{scott2009}.  He considers a financial time-series application in which $Y_i(t)$ is related to the $i^{\text{th}}$ firm's return on assets for year $t$.  \cite{scott2009} takes the $Y_i(t)$ to be the standardized residuals obtained when return on assets is regressed on several important factors.  Therefore, firms whose $Y$-process has zero mean are ordinary, those with non-zero mean are somehow extraordinary.  The goal is to flag those extraordinary firms.  \citet{scott2009} considers a nonparametric Bayes model where $f$ is sampled from a Dirichlet process distribution, and the signal trajectories are constant zero with probability $\theta$ or sampled from a Gaussian process with probability $1-\theta$.  Model \eqref{eq:armixture} is a special case in which the signal trajectories are restricted to constant functions.  But the assumption of constant paths is not critical.  Indeed, the marginal likelihood procedure described below can be extended to handle signal trajectories with a relatively low-dimensional parametrization.  

From model \eqref{eq:armixture}, it is clear that $Y_1,\ldots,Y_n$ are independent with common mixture density
\begin{align}
m_{f,\theta}(y) & = \int_{\U} \int_{-\infty}^\infty N(y \mid \xi 1, \Sigma_u) \bigl\{ \theta \langle 0 \rangle (\xi) + (1-\theta) N(\xi \mid 0,1) \bigr\} \,d\xi \, f(u) \,du  \notag \\
& = \int_{\U} \bigl\{ \theta N(y \mid 0, \Sigma_u) + (1-\theta) N(y \mid 0, \Sigma_u + 1_T 1_T') \bigr\} f(u)\,du. \label{eq:ARmix}
\end{align}
Here, with the particular choice of parametric mixing distribution, $\xi$ can be integrated out analytically, leaving only a mixture over $u$.  Predictive recursion marginal likelihood is used to estimate $\theta$ and $f$ and, in turn, these estimates are used to classify the sample paths $Y_1,\ldots,Y_n$.  

If $\theta$ and $f$ were known, the Bayes oracle rule for classifying $Y_i$ as null or non-null is as follows: for 0--1 loss, conclude $Y_i$ is non-null if 
\[ \theta_i = \frac{\theta \int N(Y_i \mid 0, \Sigma_u)f(u)\,du}{m_{f,\theta}(Y_i)} \leq 0.5. \]
Since $\theta$ and $f$ are unknown, we mimic the Bayes oracle classifier with the plug-in estimate 
\begin{equation}
\label{eq:fdr}
\hat\theta_i = \frac{\hat\theta \int N(Y_i \mid 0, \Sigma_u) f_{n,\hat\theta}(u) \,du}{m_{n, \hat\theta}(Y_i)}, 
\end{equation}
which is an estimate of the local false discovery rate \citep{efron2004, efron2008, suncai2007}, and can be viewed as an empirical Bayes approximation of the posterior inclusion probabilities in \citet{scott2009}, there obtained by Markov chain Monte Carlo.  

For illustration, we simulate 100 datasets from the model \eqref{eq:armixture}, with $n=5000$ and $T=50$, and compare the performance of our empirical Bayes classifier \eqref{eq:fdr} to the Bayes oracle.  We consider both a dense case, $\theta = 0.75$, and a sparse case, $\theta = 0.95$.  In both cases the true mixing distribution $f(u)=g_1(\phi)g_2(\sigma^2)$ is a product of shifted and scaled beta densities.  Predictive recursion marginal likelihood is applied, averaging over 25 random permutations, to estimate $\theta$ and a summary of the estimates for both the dense and sparse cases can be found in Table~\ref{table:armix}.  We find that the estimates of $\theta$ are unbiased with relatively small standard errors.  Histograms, not displayed, show roughly symmetric distributions centered at the true $\theta$.  

For the testing/classification problem, we consider two measures of performance: false discovery rate and misclassification probability.  Empirical estimates of these two quantities appear in Table~\ref{table:armix} for both the plug-in and Bayes oracle rules, and both dense and sparse cases.  The difference between false discovery rates is negligible in both the dense and sparse cases.  Furthermore, the misclassification probability for the predictive recursion-based empirical Bayes classifier is just slightly higher than that of the Bayes oracle, suggesting that the latter mimics the Bayes oracle classifier very well.

\begin{table}
\begin{center}
{\small
\begin{tabular}{cccccc}
\hline 
 & Method & Mean & SD & FDR & MP \\
\hline 
Dense, $\theta = 0.75$ & Plug-in & 0.750 & 0.008 & 0.085 & 0.109 \\
 & Oracle & --- & --- & 0.083 & 0.108 \\
Sparse, $\theta = 0.95$ & Plug-in & 0.949 & 0.004 & 0.082 & 0.026 \\
 & Oracle & --- & --- & 0.073 & 0.026 \\
\hline
\end{tabular}
}
\end{center}
\caption{Summary of inference in the autoregressive process mixture model simulations for the plug-in empirical Bayes and Bayes oracle methods. Mean denotes the simulation mean of $\hat\theta$ and SD is its standard deviation; FDR is the observed false discovery rate of the test and MP is its misclassification probability.}
\label{table:armix}
\end{table}

\section{Discussion}
\label{S:discuss}

The focus in this paper is on frequentist semiparametric inference in mixture models with a predictive recursion-based approximation to a Dirichlet process mixture marginal likelihood.  But a Bayesian version can be developed without much additional effort.  In particular, information gathered from maximizing $\lprml{n}(\theta)$ in \eqref{eq:Lprml} can be used to construct efficient importance samplers to carry out posterior computation, i.e., the maximizer and Hessian matrix of $\lprml{n}(\theta)$ can be used to construct Gaussian or heavy-tailed Student-t importance densities for $\theta$ \citep{geweke1989}. 

The choice of weights $\{w_n: n \geq 1\}$ for predictive recursion remains an important open problem.  Convergence theory gives only minimal guidelines but, in our experience, the finite sample performance is relatively robust to the choice of weights.  In this paper, we have employed the theoretically ideal weights $w_i = (i+1)^{-\gamma}$, with $\gamma = 2/3$, based on the rate in Theorem~\ref{thm:mt}.  An alternative is to let the exponent $\gamma$ be an additional tuning parameter to maximize the approximate marginal likelihood $\lprml{n}$ over \citep{taonewton1999}.  It is not yet clear, however, if the convergence theorems of \citet{mt-rate} can cover data-dependent weight sequences.  In the random-intercept regression problems of Section~\ref{SS:mixed}, the approach with estimated $\gamma$ gave similar results, not reported, to the that with fixed $\gamma=2/3$.  

The use of Hessian-based approximations of the sampling distribution of the predictive recursion marginal likelihood estimate $\hat\theta$ is based on a theory of asymptotic normality which is not yet available.  An alternative to the Hessian-based approach is a basic bootstrap.  Empirical results, not presented here, indicate that bootstrap-based confidence intervals have good coverage properties, but progress on the validity of the bootstrap for semiparametric problems has only just recently been made \citep{chenghuang2010}.  Since $\log\lprml{n}(\theta)$ is not an empirical processes, these first results do not directly apply in our context.  It is our experience that both of these approaches for approximating the sampling distribution of $\hat\theta_n$ are successful, but we leave theoretical verification of their validity to future research.

\section*{Acknowledgments}

The authors are grateful to Professor J.~K.~Ghosh for many helpful discussions, and to the Associate Editor and two referees for a number of invaluable suggestions.

\appendix

\section*{Appendix 1: Proof of Theorem~\ref{thm:emp-KL}}

Fix $\theta$ and define the sequence of random variables $Z_i = Z_i(\theta)$ as  
\[ Z_i = \log \frac{m(Y_i)}{m_{i-1,\theta}(Y_i)} - K(m,m_{i-1,\theta}), \quad i \geq 1 \]
and note that $\E(Z_i \mid \A_{i-1}) = 0$, where $\A_{i-1} = \sigma(Y_i,\ldots,Y_{i-1})$.  Therefore, $\{(Z_i,\A_i): i \geq 1\}$ forms a zero mean martingale sequence.  Next, let $m_{f,\theta}$ be the mixture density closest to $m$ in the Kullback--Leibler sense.  Then we can write 
\begin{align*}
\E(Z_i^2 \mid \A_{i-1} ) & \leq \int \Bigl\{ \log \frac{m(y)}{m_{i-1,\theta}(y)} \Bigr\}^2 m(y) \,dy \\
& = \int \Bigl\{ \log\frac{m_{f,\theta}(y)}{m_{i-1,\theta}(y)} + \log\frac{m(y)}{m_{f,\theta}(y)} \Bigr\}^2 m(y) \,dy \\
& \leq 2 \int \Bigl\{ \log\frac{m_{f,\theta}(y)}{m_{i-1,\theta}(y)} \Bigr\}^2 m(y) \,dy + 2 \int \Bigl\{ \log\frac{m(y)}{m_{f,\theta}(y)} \Bigr\}^2 m(y) \,dy \\
& = 2T_1 + 2T_2.
\end{align*}
The second term, $T_2$, is bounded by a constant $B$ according to Assumption~\ref{as:bound2}.  For the first term, let $\Y_0 = \{y: m_{f,\theta}(y) < m_{i-1,\theta}(y)\}$.  By properties of the logarithm we get 
\begin{align}
T_1 & = \int \Bigl\{ \log\frac{m_{f,\theta}(y)}{m_{i-1,\theta}(y)} \Bigr\}^2 \, m(y) \,dy \notag \\
& = \int_{\Y_0} \Bigl\{ \log\frac{m_{i-1,\theta}(y)}{m_{f,\theta}(y)} \Bigr\}^2 \, m(y) \,dy + \int_{\Y_0^c} \Bigl\{ \log\frac{m_{f,\theta}(y)}{m_{i-1,\theta}(y)} \Bigr\}^2 \, m(y) \,dy \notag \\
& \leq \int_{\Y_0} \Bigl\{ \frac{m_{i-1,\theta}(y)}{m_{f,\theta}(y)}-1 \Bigr\}^2 \, m(y) \,dy + \int_{\Y_0^c} \Bigl\{ \frac{m_{f,\theta}(y)}{m_{i-1,\theta}(y)} -1 \Bigr\}^2 \, m(y) \,dy \notag \\
& \leq 2 + \int \Bigl[ \Bigl\{ \frac{m_{i-1,\theta}(y)}{m_{f,\theta}(y)} \Bigr\}^2 + \Bigl\{ \frac{m_{f,\theta}(y)}{m_{i-1,\theta}(y)} \Bigr\}^2 \Bigr] m(y) \,dy \notag \\
& \leq 2 + 2 \sup_{u_1,u_2} \int \Bigl\{ \frac{p(y|u_1,\theta)}{p(y|u_2,\theta)} \Bigr\}^2 m(y) \,dy, \label{eq:jensen} 
\end{align}
where \eqref{eq:jensen} follows by two applications of Jensen's inequality, one to the mapping $x \mapsto 1/x$ and one to the mapping $x \mapsto x^2$.  The last term is bounded according to Assumption~\ref{as:bound1}.  Therefore, $\E(Z_i^2 \mid \A_{i-1})$ is uniformly bounded by a constant $M$ and, consequently, $v_n = \sum_{i=1}^n \E( Z_i^2 \mid \A_{i-1} ) \leq Mn$.  

Set $b_n = n/c_n$, where $c_n = a_n$ or $c_n \equiv 1$, depending on whether or not the conditions of the second part of the theorem hold.  It is clear that $b_n$ grows faster than $n^{1/2}$, but no faster than $n$, which implies 
\begin{equation}
\label{eq:teicher1}
\frac{v_n^{1/2}}{b_n (\log\log b_n)^{-1/2}} \leq \frac{(M n \log \log b_n)^{1/2}}{b_n} \to 0. 
\end{equation}
Furthermore, by Markov's inequality, we have, with probability 1, 
\begin{equation}
\label{eq:teicher2}
\sum_{n=1}^\infty \prob\Bigl(|Z_n| > \frac{b_n}{\log\log b_n} \, \Bigl\lvert \, \A_{n-1} \Bigr) \leq M \sum_{n=1}^\infty \frac{(\log\log b_n)^2}{b_n^2} < \infty. 
\end{equation}
In light of \eqref{eq:teicher1} and \eqref{eq:teicher2}, it now follows from Corollary~2 of \citet{teicher98} that $b_n^{-1}\sum_{i=1}^n Z_i \to 0$ almost surely.  Therefore, we can conclude that, with probability 1, 
\begin{align}
c_n \Bigl\lvert &K_n(\theta) - \frac1n \sum_{i=1}^n K(m,m_{i-1,\theta}) \Bigr\rvert\notag \\
& = \Bigl\lvert c_n\bigl\{ K_n(\theta) - K^\star(\theta) \bigr\} - \frac{c_n}{n} \sum_{i=1}^n \bigl\{ K(m,m_{i-1,\theta})-K^\star(\theta) \bigr\} \Bigr\rvert \to 0.  \label{eq:emp-KL}
\end{align}
If $c_n \equiv 1$, then the result follows from Theorem~\ref{thm:mt} and Cesaro's theorem.  If $c_n = a_n$, write $\kappa_i = K(m,m_{i-1,\theta})-K^\star(\theta)$ and note that $\kappa_i \geq 0$.  Summation-by-parts and monotonicity of the weights $w_i$ yield the following inequality:
\begin{equation}
\label{eq:cesaro}
a_n \sum_{i=1}^n \kappa_i = \sum_{i=1}^n a_i \kappa_i + \sum_{i=1}^{n-1} w_{i+1} \Bigl( \sum_{j=1}^i \kappa_j \Bigr) \leq \sum_{i=1}^n a_i \kappa_i + \sum_{i=1}^n \Bigl( \sum_{j=1}^i w_j \kappa_j \Bigr). 
\end{equation}
Let $S_i = \sum_{j=1}^i w_j \kappa_j$ for $i \geq 1$.  Lemma~4.4 of \citet{mt-rate} shows that the limit $S_\infty$ is finite almost surely.  Dividing through \eqref{eq:cesaro} by $n$ and applying Cesaro's theorem shows that the right-hand side is positive and bounded by $S_\infty$ almost surely for large $n$.  This observation, together with \eqref{eq:emp-KL}, implies $a_n \{K_n(\theta) - K^\star(\theta)\}$ is almost surely bounded, proving the theorem.

\section*{Appendix 2: Predictive recursion gradient algorithm}

Here we present a version of the predictive recursion that gives the gradient of $\ell_n(\theta) = \log\lprml{n}(\theta)$ as a by-product.  Let $\lambda_i(\theta) = m_{i-1,\theta}(Y_i)$; then $\grad \ell_n(\theta) = \sum_{i=1}^n \grad \log \lambda_i(\theta)$.  For a function $g(\theta,u)$, the notation $\grad g(\theta,u)$ means the gradient of $g$ with respect to $\theta$, pointwise in $u$.  

As in the original algorithm, the user must be able to evaluate the kernel $p(Y_i \mid \theta,u)$ at each $Y_i$ for any pair $(\theta,u)$.  Furthermore, for this modification, the user must also be able to evaluate $\grad p(Y_i \mid \theta,u)$.  

\begin{enumerate}

\item Start with the user-defined $f_{0,\theta}(u)$ and compute $\grad f_{0,\theta}(u)$.  

\item For $i=1,\ldots,n$, repeat the following three steps:

\begin{enumerate}

\item Set $g(\theta,u) = p(Y_i \mid \theta,u)$, $\grad g(\theta,u) = \grad p(Y_i \mid \theta,u)$, and 
\[ G(\theta,u) = g(\theta,u) \grad f_{i-1,\theta}(u) + \grad g(\theta,u) f_{i-1,\theta}(u). \]

\item Compute $\lambda_i(\theta) = \int g(\theta,u) f_{i-1,\theta}(u) \,d\mu(u)$ and $\grad \log \lambda_i(\theta) = \int G(\theta,u) \,d\mu(u) / \lambda_i(\theta)$.

\item Update 
\begin{align*}
f_{i,\theta}(u) & = (1-w_i) f_{i-1,\theta}(u) + w_i \frac{g(\theta,u) f_{i-1,\theta}(u)}{\lambda_i(\theta)}, \\
\grad f_{i,\theta}(u) & = (1-w_i)\grad f_{i-1,\theta}(u) + w_i \Bigl\{ \frac{G(\theta,u) - g(\theta,u) f_{i-1,\theta}(u) \grad \log \lambda_i(\theta)}{\lambda_i(\theta)} \Bigr\}.
\end{align*}

\end{enumerate}

\item Return $f_{n,\theta}(u)$, $m_{n,\theta}(y)$, and $\sum_{i=1}^n \grad\log \lambda_i(\theta)$.  

\end{enumerate}

\bibliographystyle{/home/rgmartin/Research/TexStuff/asa}
\bibliography{/home/rgmartin/Research/mybib}

\begin{thebibliography}{25}
\newcommand{\enquote}[1]{``#1''}
\expandafter\ifx\csname natexlab\endcsname\relax\def\natexlab#1{#1}\fi

\bibitem[{Antoniak(1974)}]{antoniak1974}
Antoniak, C. (1974), \enquote{Mixtures of {D}irichlet processes with
  applications to {B}ayesian nonparametric problems,} \textit{Ann. Statist.},
  2, 1152--1174.

\bibitem[{Blackwell and MacQueen(1973)}]{blackwellmacqueen}
Blackwell, D. and MacQueen, J.~B. (1973), \enquote{Ferguson distributions via
  {P}\'olya urn schemes,} \textit{Ann. Statist.}, 1, 353--355.

\bibitem[{Bush and MacEachern(1996)}]{bush.maceachern.1996}
Bush, C.~A. and MacEachern, S.~N. (1996), \enquote{A semiparametric {B}ayesian
  model for randomised block designs,} \textit{Biometrika}, 83, 275--285.

\bibitem[{Cheng and Huang(2010)}]{chenghuang2010}
Cheng, G. and Huang, J. (2010), \enquote{Bootstrap consistency for general
  semiparametric {M}-estimation,} \textit{Ann. Statist.}, 38, 2884--2915.

\bibitem[{Efron(2004)}]{efron2004}
Efron, B. (2004), \enquote{Large-scale simultaneous hypothesis testing: the
  choice of a null hypothesis,} \textit{J. Amer. Statist. Assoc.}, 99, 96--104.

\bibitem[{Efron(2008)}]{efron2008}
--- (2008), \enquote{Microarrays, empirical {B}ayes and the two-groups model,}
  \textit{Statist. Sci.}, 23, 1--22.

\bibitem[{Ferguson(1973)}]{ferguson1973}
Ferguson, T.~S. (1973), \enquote{A {B}ayesian analysis of some nonparametric
  problems,} \textit{Ann. Statist.}, 1, 209--230.

\bibitem[{Geweke(1989)}]{geweke1989}
Geweke, J. (1989), \enquote{Bayesian inference in econometric models using
  {M}onte {C}arlo integration,} \textit{Econometrica}, 57, 1317--1339.

\bibitem[{Ghosh and Ramamoorthi(2003)}]{ghoshramamoorthi}
Ghosh, J.~K. and Ramamoorthi, R.~V. (2003), \textit{Bayesian {N}onparametrics},
  New York: Springer-Verlag.

\bibitem[{Ghosh and Tokdar(2006)}]{ghoshtokdar}
Ghosh, J.~K. and Tokdar, S.~T. (2006), \enquote{Convergence and consistency of
  {N}ewton's algorithm for estimating mixing distribution,} in
  \textit{Frontiers in statistics}, eds. Fan, J. and Koul, H., London: Imp.
  Coll. Press, pp. 429--443.

\bibitem[{Jefferys and Berger(1992)}]{jefferys}
Jefferys, W. and Berger, J. (1992), \enquote{Ockham's razor and {B}ayesian
  analysis,} \textit{American Scientist}, 80, 64--72.

\bibitem[{MacEachern et~al.(1999)MacEachern, Clyde, and
  Liu}]{maceachernclydeliu}
MacEachern, S.~N., Clyde, M., and Liu, J.~S. (1999), \enquote{Sequential
  importance sampling for nonparametric {B}ayes models: the next generation,}
  \textit{Canad. J. Statist.}, 27, 251--267.

\bibitem[{Martin and Ghosh(2008)}]{martinghosh}
Martin, R. and Ghosh, J.~K. (2008), \enquote{Stochastic approximation and
  {N}ewton's estimate of a mixing distribution,} \textit{Statist. Sci.}, 23,
  365--382.

\bibitem[{Martin and Tokdar(2009)}]{mt-rate}
Martin, R. and Tokdar, S.~T. (2009), \enquote{Asymptotic properties of
  predictive recursion: robustness and rate of convergence,} \textit{Electron.
  J. Stat.}, 3, 1455--1472.

\bibitem[{Newton(2002)}]{newton02}
Newton, M.~A. (2002), \enquote{On a nonparametric recursive estimator of the
  mixing distribution,} \textit{Sankhy\=a Ser. A}, 64, 306--322.

\bibitem[{Newton et~al.(1998)Newton, Quintana, and Zhang}]{nqz}
Newton, M.~A., Quintana, F.~A., and Zhang, Y. (1998), \enquote{Nonparametric
  {B}ayes methods using predictive updating,} in \textit{Practical
  nonparametric and semiparametric Bayesian statistics}, eds. Dey, D.,
  M\"uller, P., and Sinha, D., New York: Springer, vol. 133 of \textit{Lecture
  Notes in Statist.}, pp. 45--61.

\bibitem[{Newton and Zhang(1999)}]{newtonzhang}
Newton, M.~A. and Zhang, Y. (1999), \enquote{A recursive algorithm for
  nonparametric analysis with missing data,} \textit{Biometrika}, 86, 15--26.

\bibitem[{Quintana and Newton(2000)}]{quintana2000}
Quintana, F.~A. and Newton, M.~A. (2000), \enquote{Computational aspects of
  nonparametric {B}ayesian analysis with applications to the modeling of
  multiple binary sequences,} \textit{J. Comput. Graph. Statist.}, 9, 711--737.

\bibitem[{Scott(2009)}]{scott2009}
Scott, J.~G. (2009), \enquote{Nonparametric {B}ayesian multiple testing for
  longitudinal performance stratification,} \textit{Ann. Appl. Statist.}, 3,
  1655--1674.

\bibitem[{Sun and Cai(2007)}]{suncai2007}
Sun, W. and Cai, T.~T. (2007), \enquote{Oracle and adaptive compound decision
  rules for false discovery rate control,} \textit{J. Amer. Statist. Assoc.},
  102, 901--912.

\bibitem[{Tao et~al.(1999)Tao, Palta, Yandell, and Newton}]{taonewton1999}
Tao, H., Palta, M., Yandell, B.~S., and Newton, M.~A. (1999), \enquote{An
  estimation method for the semiparametric mixed effects model,}
  \textit{Biometrics}, 55, 102--110.

\bibitem[{Teicher(1998)}]{teicher98}
Teicher, H. (1998), \enquote{Strong laws for martingale differences and
  independent random variables,} \textit{J. Theoret. Probab.}, 11, 979--995.

\bibitem[{Tierney and Kadane(1986)}]{tierney.kadane.1986}
Tierney, L. and Kadane, J.~B. (1986), \enquote{Accurate approximations for
  posterior moments and marginal densities,} \textit{J. Amer. Statist. Assoc.},
  81, 82--86.

\bibitem[{Tokdar et~al.(2009)Tokdar, Martin, and Ghosh}]{tmg}
Tokdar, S.~T., Martin, R., and Ghosh, J.~K. (2009), \enquote{Consistency of a
  recursive estimate of mixing distributions,} \textit{Ann. Statist.}, 37,
  2502--2522.

\bibitem[{Wald(1949)}]{wald1949}
Wald, A. (1949), \enquote{Note on the consistency of the maximum likelihood
  estimate,} \textit{Ann. Math. Statist.}, 20, 595--601.

\end{thebibliography}

\end{document}